\documentclass[10pt, conference, letterpaper]{IEEEtran}
\IEEEoverridecommandlockouts
\usepackage{cite}
\usepackage{amsmath,amssymb,amsfonts}
\usepackage{graphicx}
\usepackage{textcomp}
\usepackage{xcolor}
\usepackage{amsthm}
\def\BibTeX{{\rm B\kern-.05em{\sc i\kern-.025em b}\kern-.08em
    T\kern-.1667em\lower.7ex\hbox{E}\kern-.125emX}}
    
\newtheorem{theorem}{Theorem}
\newtheorem{definition}{Definition}
\newtheorem{problem}{Problem}

\usepackage{algorithm,algorithmicx,algpseudocode}
\usepackage{enumitem}
\usepackage{subfigure}

\usepackage{url}

\newcommand{\ie}{\textit{i}.\textit{e}., }

\newcommand\Tstrut{\rule{0pt}{2.4ex}}         
\newcommand\Bstrut{\rule[-1.0ex]{0pt}{0pt}}   
\newcommand\TBstrut{\Tstrut\Bstrut}           
\usepackage{siunitx}

\usepackage{tablefootnote}
\usepackage{multirow}
\begin{document}

\title{Federated Analytics-Empowered Frequent Pattern Mining for Decentralized Web 3.0 Applications\\
\thanks{This work is supported by the National Natural Science Foundation of China (Grant No. 62302292), the National Key R\&D Program of China (Grant No. 2023YFB2704400), and the Fundamental Research Funds for the Central Universities.
The work of D. Wang is supported by RGC-GRF 15209220, 15200321, 15201322, from ITC via project ”Smart Railway Technology and Applications” (No. K-BBY1), RGC-CRF C5018-20G, ITC ITF-ITS/056/22MX.
The work of Z. Han is supported by NSF CNS-2107216, CNS-2128368, CMMI-2222810, ECCS-2302469, US Department of Transportation, Toyota and Amazon. Corresponding author: Yifei Zhu.
}
}

\author{\IEEEauthorblockN{Zibo Wang\IEEEauthorrefmark{1},
Yifei Zhu\IEEEauthorrefmark{1}\IEEEauthorrefmark{3},
Dan Wang\IEEEauthorrefmark{4}, and
Zhu Han\IEEEauthorrefmark{5}}
\IEEEauthorblockA{\IEEEauthorrefmark{1}UM-SJTU Joint Institute, Shanghai Jiao Tong University}
\IEEEauthorblockA{\IEEEauthorrefmark{3}Cooperative Medianet Innovation Center (CMIC), Shanghai Jiao Tong University}
\IEEEauthorblockA{\IEEEauthorrefmark{4}Department of Computing, The Hong Kong Polytechnic University}
\IEEEauthorblockA{\IEEEauthorrefmark{5}Department of Electrical and Computer Engineering, University of Houston
\\Email: wangzibo@sjtu.edu.cn, yifei.zhu@sjtu.edu.cn, csdwang@comp.polyu.edu.hk, zhan2@uh.edu}
}

\maketitle

\begin{abstract}

The emerging Web 3.0 paradigm aims to decentralize existing web services, enabling desirable properties such as transparency, incentives, and privacy preservation. However, current Web 3.0 applications supported by blockchain infrastructure still cannot support complex data analytics tasks in a scalable and privacy-preserving way. This paper introduces the emerging federated analytics (FA) paradigm into the realm of Web 3.0 services, enabling data to stay local while still contributing to complex web analytics tasks in a privacy-preserving way.  We propose FedWeb, a tailored FA design for important frequent pattern mining tasks in Web 3.0. FedWeb remarkably reduces the number of required participating data owners to support privacy-preserving Web 3.0 data analytics based on a novel distributed differential privacy technique. The correctness of mining results is guaranteed by a theoretically rigid candidate filtering scheme based on Hoeffding's inequality and Chebychev's inequality. Two response budget saving solutions are proposed to further reduce participating data owners. Experiments on three representative Web 3.0 scenarios show that FedWeb can improve data utility by $\sim$25.3\% and reduce the participating data owners by $\sim$98.4\%.

\end{abstract}


\section{Introduction}

\textbf{When federated analytics meets Web 3.0.} Web 3.0 is believed to be the next generation of the Internet.
Compared to previous Web 2.0 driven by dominant Internet application providers (sometimes termed \textit{tech giants}), Web 3.0 starts a revolution in the storage, sharing, and profiting of personal data \cite{web3jp,web3eth}. In Web 2.0, these data are collected by the applications and utilized to create profit for the tech giants, resulting in privacy leakage and unfairness in data profiting. Web 3.0 advocates are wishing to resolve the issues by the decentralized nature of Web 3.0. In Web 3.0, only the user (data owner) can access and modify the private data, because these data are stored in the user devices, or encrypted by the users. Other entities can access the private data only when they are permitted by the data owner. In this setting, user privacy is properly handled, and the strongest privacy can be achieved when the data owner never grants permission to the data. In 2022, the market size of Web 3.0 reached \$1.73 billion \cite{web3market}.

\newcommand{\configheight}{0.17\textwidth}
\begin{figure*}[!tbp]
    \centering
    \subfigure[Complete data hiding sacrifices analytics]
    {
        \includegraphics[height=\configheight]{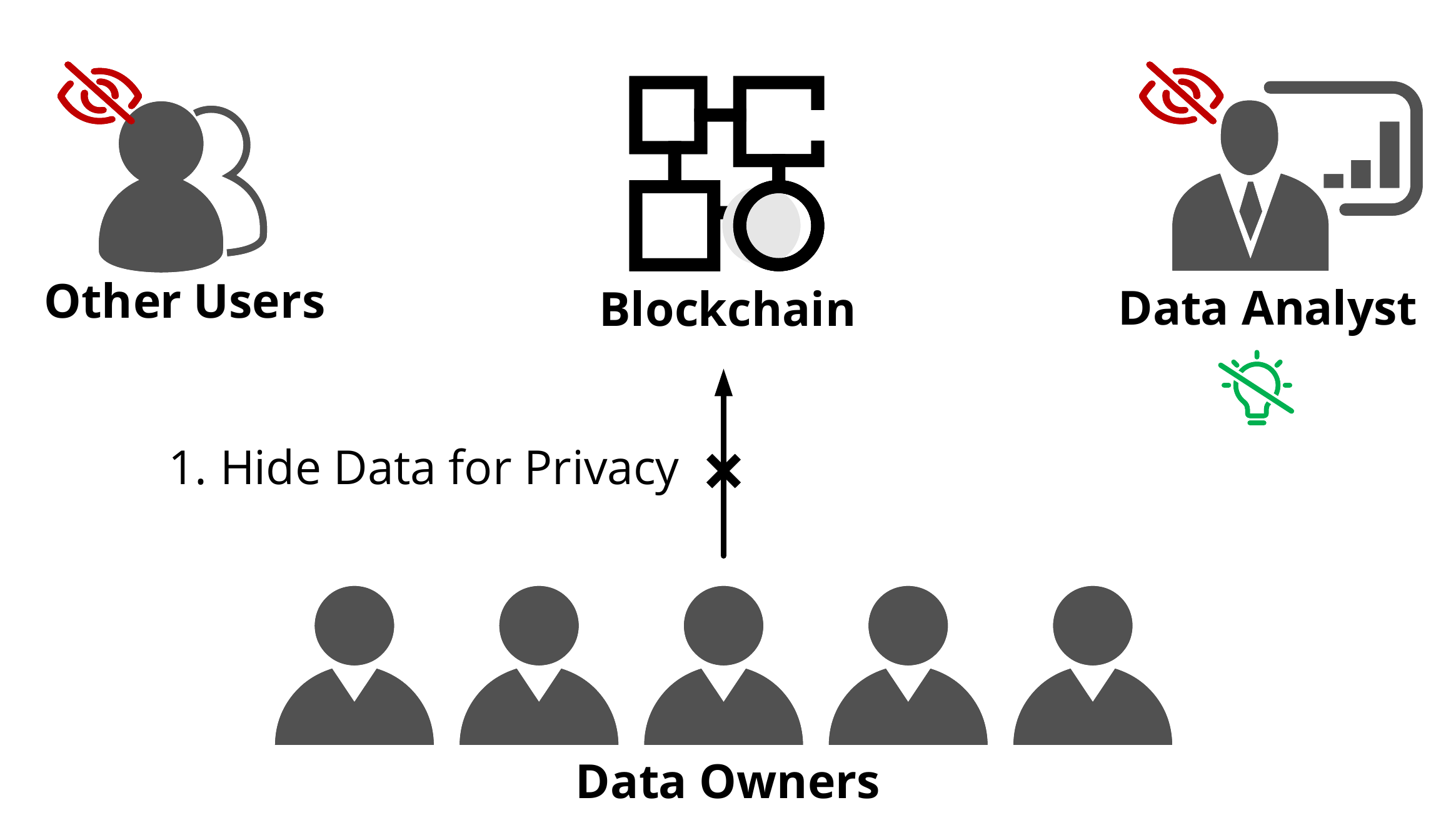}
        \label{subfig_concept_1}
    }
    \subfigure[Direct data sharing sacrifices privacy]
    {
        \includegraphics[height=\configheight]{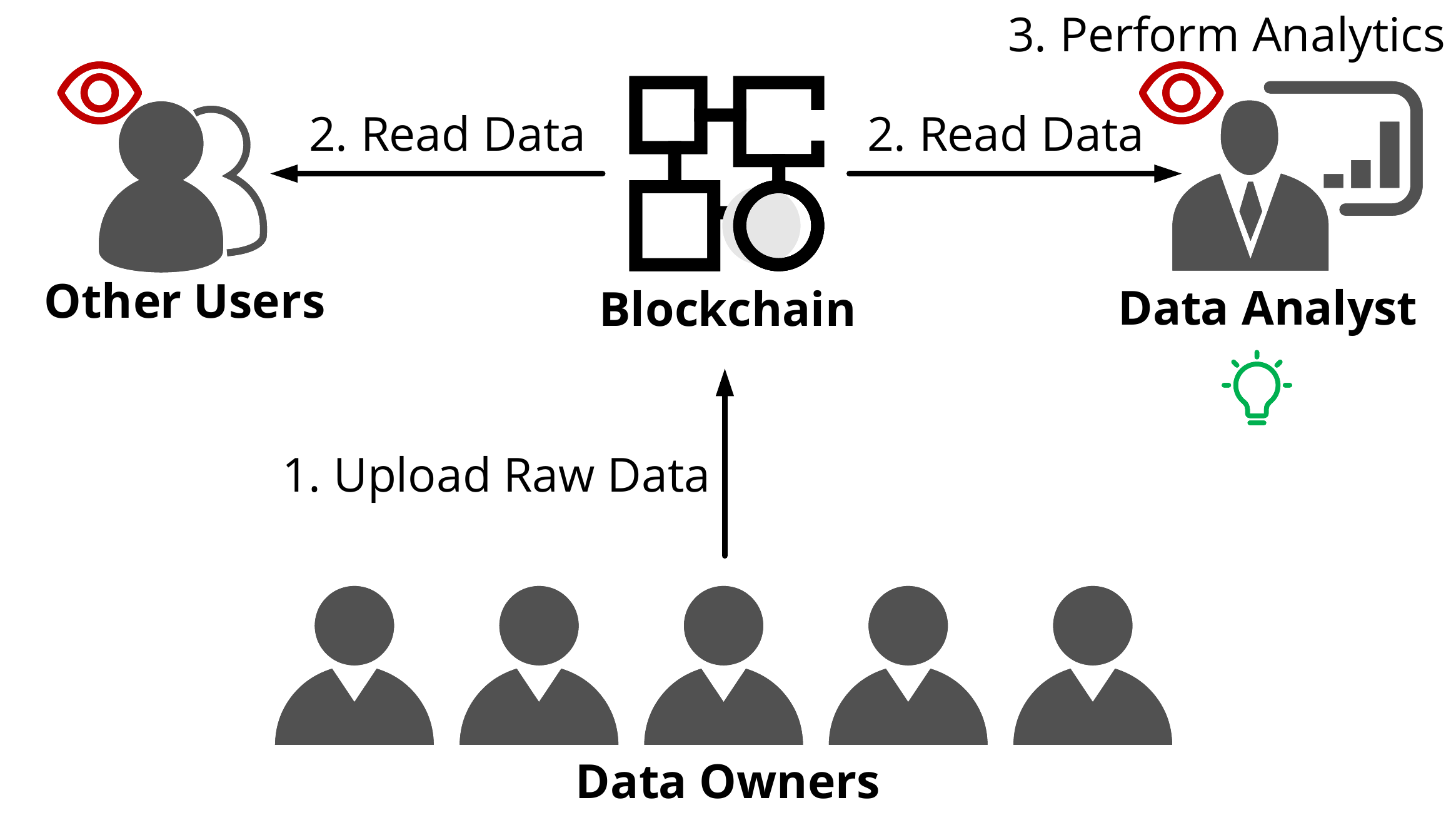}
        \label{subfig_concept_2}
    }
    \subfigure[FA-based data analytics for Web 3.0]
    {
        \includegraphics[height=\configheight]{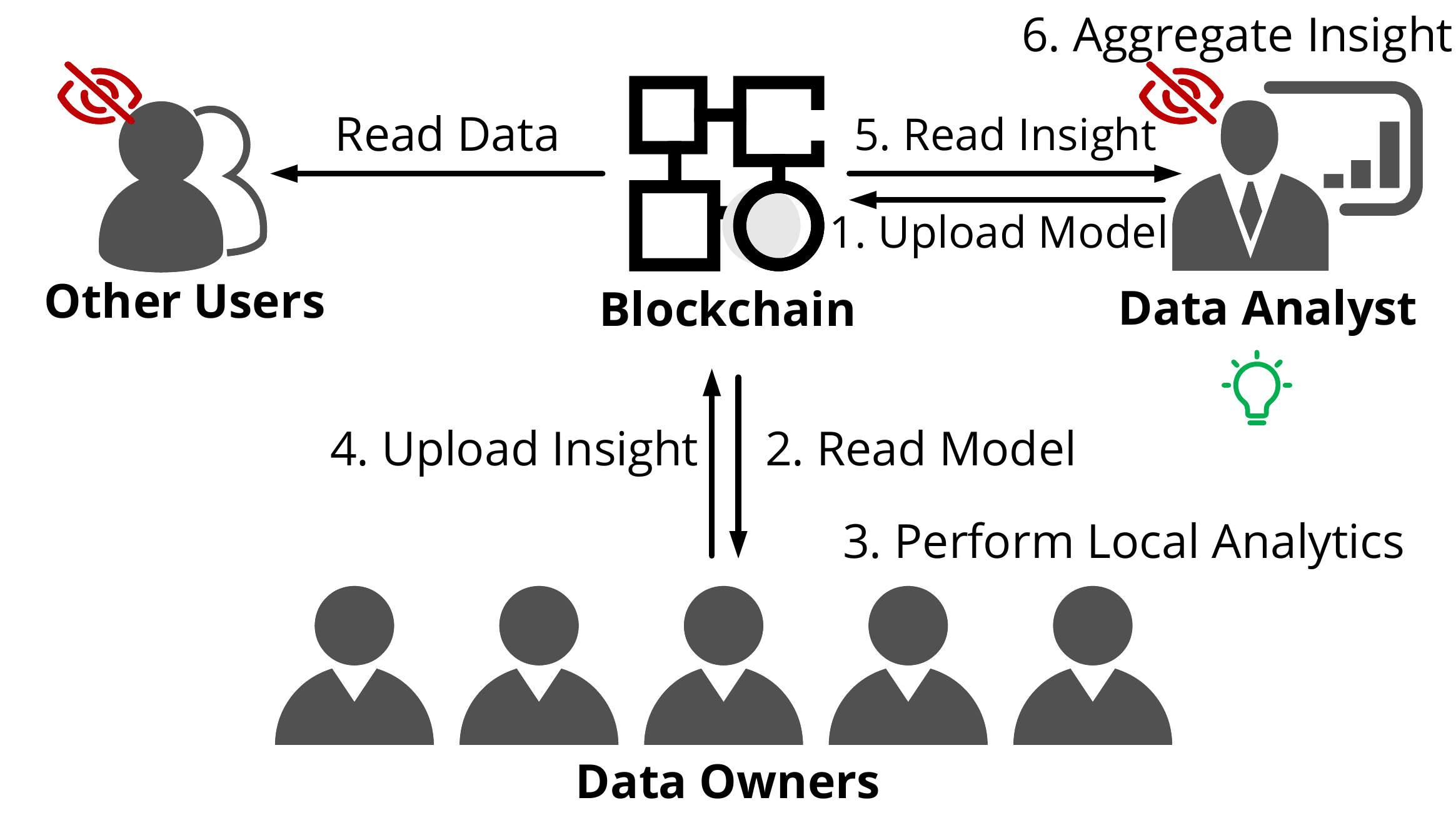}
        \label{subfig_concept_3}
    }
    \newline
    \caption{Possible statuses of privacy preservation and data analytics in Web 3.0 applications. The marks of {\color{red!70!black}eye} indicate the availability to learn private data. The marks of {\color{green!70!black}lamp} indicate the availability to optimize the web service via data analytics.}
    \label{fig_concept}
\end{figure*}

\begin{table*}
\caption{Web 3.0 services, their FPM-related applications, and the corresponding FPM tasks.}
\label{table_app}
\centering
\begin{tabular}{lllll}
\hline
  \textbf{Web service}& \textbf{Web 3.0 example}& \textbf{FPM-related application} & \textbf{Mined data} & \textbf{Pattern type} \TBstrut\\
\hline

    Internet browser & Brave \cite{brave} & User traversal analysis \cite{sun2004efficient} & Web browsing history & Sequence \Tstrut\\
    Operating system & WRIO \cite{wrio} & Next used app prediction \cite{lu2014mining} & Mobile app usage logs & Sequence\\
    Social network & Steemit \cite{steemit} & Trend mining \cite{nohuddin2012finding} & Timestamped social media posts & Itemset\\
    Online market &  OpenBazaar \cite{openbazaar}  & Market basket analysis \cite{saputra2023market} & Consumer baskets & Itemset\\
    Computer security & Web3 Antivirus \cite{web3antivirus} & Ransomware detection \cite{homayoun2017know} & System logs & Sequence\\
    Micro-blogging & Mirror \cite{mirror} & Emotion analysis \cite{skenduli2021mining} & Emotion classfication results of posts & Itemset\\
    Media player & OPUS \cite{opus} & Music recommendation \cite{hariri2012context} & Historical music playlists & Sequence\Bstrut\\ 
    
\hline
\end{tabular}  
\end{table*}

Although Web 3.0 has utilized its blockchain infrastructure to accomplish transparency, auditability, and data ownership management, privacy-preserving web data analytics is still a missing piece to complete the Web 3.0 ecosystem. 
In modern web applications, data analytics acts an important role in many tasks, like recommendation systems and service monitoring \cite{han2022data}. Analyzing data from vast data owners, if not misused, benefit both service providers and service users. These data analytics applications undoubtedly will continue to be a cornerstone for Web 3.0 applications.

However, to pursue the strongest privacy preservation, data owners may hide their sensitive data. It disables the data analytics and consequently degrades the web service and hurts both the service providers and users, as demonstrated in Fig. \ref{subfig_concept_1}. Alternatively, data owners can share their data by providing data access to the data analysts. But it removes the privacy preservation of the Web 3.0 scheme, as demonstrated in Fig. \ref{subfig_concept_2}.

Federated analytics (FA), a paradigm for privacy-preserving data analytics, has gained great success recently\cite{wang2021federated}. 
In FA, the raw data are prohibited to be transmitted to any other entities, which preserves the data privacy of data owners to some extent \cite{fa20,fltutorial}. In addition, it is able to apply additional privatization techniques, like cryptography or differential privacy (DP), to obtain a stronger and more formal privacy guarantee. 
FA has been applied in various data-driven scenarios, where data privacy is concerned by user awareness and regulations \cite{zhu2020federated,dennis2021heterogeneity,cormode2021bit,wang2022fedfpm}. 
The decentralized nature of FA naturally matches the requirement of Web 3.0 applications. 

To fill in the gap between the requirement for privacy preservation and the demand for high quality web data analytics in Web 3.0, in this paper, we introduce FA to the realm of Web 3.0 services. 
In our design, the data owners still grant data analysts permission to the local data. However, the data owners no longer let the data analysts access their raw data. \textit{Instead, they offer FA services to the data analysts, so that the data analytics tasks can be properly completed, while the data privacy is preserved}. The FA-assisted Web 3.0 data analytics model is illustrated in Fig. \ref{subfig_concept_3}. The local privacy preservation of FA makes data owners safe to share their data, while the careful design of FA achieves the high quality of data analytics and downstream web services.
Since FA algorithms are mostly task-specific, to support various data analytics demands, different tailored FA algorithms should be designed and deployed.

\textbf{FA-based frequent pattern mining for Web 3.0 applications.} Frequent pattern mining (FPM) is one of the most important tasks that discovers all substructures (items, subsets, or subsequences) with occurrence frequencies higher than a predefined threshold among the local data of data owners. It is applied in discovering and analyzing frequently occurring web events or objects. It has wide potential applications in Web 3.0 with some examples summarized in Table \ref{table_app}.

As a classical and important data analytics task, many FA-based FPM solutions have been proposed \cite{erlingsson2014rappor,apple2017learning,qin2016heavy,wang2018locally,wang2022fedfpm}. 
Although these solutions manage to conduct FPM with privacy preservation, applying them to Web 3.0 encounters extra difficulties. One of the most critical challenges introduced by the Web 3.0 scenarios is the limited scale of participating data owners. 
Existing FA FPM solutions require tens of millions of participating data owners, and otherwise suffer from low data utility. Such requirements can be fulfilled in Web 2.0 applications as tech giants can gather mass participants. For example, Gboard (5 billion+ downloads in GooglePlay) users are utilized by Google as FA participants \cite{fa20}. However, since the early-stage Web 3.0 applications have a relatively small scale of users (around millions, taking Steemit as an example \cite{steemit,li2021steemops}), this makes the existing solutions no longer valid. 
New FA algorithms are needed to handle FPM with limited participating data owners.

In this paper, we follow the aforementioned model to apply FA in Web 3.0 ecosystems, and propose a novel FPM solution, named FedWeb, to fulfill the need for FPM in Web 3.0 services. It leverages the interactive FA structure to formulate a query-response scheme between the data analyst and data owners, and the iterative scheme to progressively update the candidate patterns with the confidence bounding and Apriori mechanism. To reduce the required participating data owners as restricted by Web 3.0, FedWeb utilizes the novel distributed differential privacy (DDP) techniques, which provide a formal privacy guarantee and reduce the noise added to the uploads. In addition, two flexible budget saving strategies, named candidate padding and data owner reusing, are proposed to further reduce the scale of participating data owners.

\textbf{Our contributions.} In summary, our contributions are:
\begin{itemize}
    \item We introduce FA into the current Web 3.0 ecosystem, so that data analytics can be properly conducted without exposing the sensitive raw data.
    \item We design FedWeb for privacy-preserving FPM in Web 3.0 services. We incorporate the DDP privatization scheme and design two strategies to save the response budgets, overcoming the scale limitation in existing Web 3.0 services (Section \ref{sec_method}).  
    \item We rigorously prove the correctness of our candidate pattern filtering scheme. It formulates two confidence bounds using compounds from the Chebyshev's inequality and the Hoeffding's inequality, providing the theoretical guarantee of the FPM results of FedWeb (Section \ref{sec_method}).
    \item Experiment results on various web service scenarios reveal that FedWeb obtains a high F1 score and requires few participating data owners. It verifies the capability of FedWeb in completing Web 3.0 FPM tasks with high quality under limited available data owners (Section \ref{sec_evaluation}).
\end{itemize}

\section{Preliminary: Privacy for Web 3.0}\label{sec_preliminary}

In this section, we discuss some concepts regarding privacy criteria which are related to our design.

DP is a gold standard in the field of private data publication, owing to its assumption of strong adversary knowledge and rigid statistical guarantee \cite{dwork2006calibrating}. Central differential privacy (CDP) is the original and most widely used version of DP. CDP typically works in scenarios where the processor who runs the data publication mechanism can access the whole database. Therefore, CDP is no longer valid in Web 3.0 scenarios, where the data should be privatized before being gathered.
Local differential privacy (LDP) are applied  to handle the local privacy case \cite{erlingsson2014rappor,wang2022fedfpm,apple2017learning}. However, LDP requires data owners to add significant noise to the uploads, which degrades the data utility and requires more data owners.

DDP is an emerging solution to provide (Web 3.0 available) local privacy preservation, while gaining high data utility equivalent to CDP \cite{bagdasaryan2022towards}. DDP is realized via two key designs: 
\begin{itemize}
    \item \textbf{Distributed noise generation}: Each data owner adds a tailored noise to its local analytics result, where the noises from all data owners sum up to satisfy CDP.
    \item \textbf{Secure aggregation}: Data owners encrypt the upload of each individual data owner, and the data analyst then can only learn the sum of data owner uploads.
\end{itemize}

DDP inherits the privacy preservation of DP owing to the ingenious combination of the two designs. Distributed noise generation guarantees that the data statistically satisfy CDP after being aggregated. Secure aggregation prevents the data analyst from learning unaggregated uploads by transforming them into meaningless random values with encryption, as the added noise is not enough to solely satisfy DP (only LDP-level noise can preserve privacy for individual uploads).

Our distributed noise generation scheme is based on the geometric mechanism \cite{goryczka2015comprehensive}, which is described as follows.
\begin{definition}[Geometric mechanism]\label{def_geo}
Given a parameter $\alpha\in(0,1)$, a $\alpha$-geometric mechanism $M$ is a data publication mechanism possessing on dataset $S$ that has integer output. Given the true query result $Q(S)$, $M$ outputs $M(S)=Q(S)+G(\alpha)$, where $G(\alpha)$ is an integer two-sided geometric random variable with the following PDF:
\begin{equation}\label{eq_geo}
    \mathbb{P}(G(\alpha)=x)=\frac{1-\alpha}{1+\alpha}\alpha^{|x|}.
\end{equation}
\end{definition}

When $Q$ is defined as summation in $S$ and the value of every record in $S$ is binary, the \textit{sensitivity} of $Q$ becomes 1. $\alpha$-geometric mechanism in this case satisfies $(-\ln\alpha)$-CDP.

To utilize the geometric mechanism in distributed noise generation, $G$ should be decomposed into multiple random variables that are added by different users respectively. P\'olya random variables are utilized to realize it, shown as follows.
\begin{definition}[P\'olya distribution]\label{def_polya}
Let $\mathcal{X}$ be a random variable following $P\acute{o}lya(r,p)$ distribution with $r\in\mathbb{R}$ and $p\in(0,1)$, its PDF is given as follows.
\begin{equation}
    \mathbb{P}(\mathcal{X}=x)=\binom{r+x-1}{r-1}p^x (1-p)^r.
\end{equation}
\end{definition}

\begin{theorem}\label{theorem_ddp}
Suppose there are $n$ distributed data owners, each data owner adds $\mathcal{X}_i - \mathcal{Y}_i$ to its upload, where $\mathcal{X}_i$ and $\mathcal{Y}_i$ are i.i.d. $P\acute{o}lya(1/n,\alpha)$ variables. These variables sum up to follow two-sided geometric distribution, \ie
\begin{equation}
    \sum_{i=1}^n (\mathcal{X}_i - \mathcal{Y}_i) = G(\alpha).
\end{equation}
\end{theorem}
\begin{proof}
See Theorem 5.1 of \cite{goryczka2015comprehensive}.
\end{proof}
Definition \ref{def_geo} and Theorem \ref{theorem_ddp} show that, by individually adding the difference of two P\'olya variables, data owners can collaboratively generate a Geometric noise on the aggregated data, which satisfies the requirement of DDP. A $P\acute{o}lya(r,p)$ variable $\mathcal{X}$ is generated via a Poisson-Gamma mixture \cite{johnson2005univariate}: first draw $\gamma$ from the $Gamma(r,p/(1-p))$ distribution, and then draw $\mathcal{X}$ from the Poisson distribution with parameter $\gamma$.

\section{Threat Model and Problem Formulation}\label{sec_system}
\textbf{Threat model.} Being consistent with many previous studies \cite{wang2022fedfpm,erlingsson2014rappor,apple2017learning,qin2016heavy,wang2018locally}, we consider the data analyst and data owners to be semi-honest, \ie the data analyst or each data owner executes the protocol honestly, but try to learn the data from the (other) data owners with the information they received during the protocol. In addition, we consider that a limited number of participants may collude to learn the data. Particularly, we assume there exists collusion between, at most, the data analyst and one-third of participating data owners.

\textbf{Problem formulation.}
Our system consists of a data analyst hosting an FPM task, and a number of data owners $D$ holding private data. The data analyst wishes to discover the frequent patterns within the local data of data owners, \ie the patterns whose frequencies are higher than a user-defined frequency threshold $f$. The FA design consists of two functions: an insight derivation function $\mathcal{I}$ which executes on one data owner $d\in D$, and an aggregation function $\mathcal{A}$ which aggregates the insights to derive the final output\footnote{We use $d$ to denote both a data owner and its local data}. 
The problem of privacy-preserving FPM in Web 3.0 scenarios we try to resolve can be described as follows.
\begin{problem}
Design the functions $\mathcal{I}$ and $\mathcal{A}$, where $\mathcal{I}$ executes on each participating data owner $d\in D'$, where the set of participating data owners $D'$ is a subset of available data owners, \ie $D'\subseteq D$. $\mathcal{A}$ takes the outputs of $\mathcal{I}$ as input, and outputs a set of discovered frequent patterns $\mathcal{F} = \mathcal{A}(\{\mathcal{I}(d) | d\in D'\})$. The algorithm design should achieve three goals:
\begin{itemize}
    \item $\mathcal{F}$ should be closed to the ground truth frequent patterns $\mathcal{F}'$, \ie high F1 score $F1(\mathcal{F},\mathcal{F}')$ should be achieved.\footnote{F1 score is defined by $2pr/(p+r)$, where $p$ is precision and $r$ is recall.}
    \item The algorithm should execute on a small scale of participating data owners, \ie the size of $D'$ should be low.
    \item The output of $\mathcal{I}$ should satisfy DDP. In detail, the upload of individual data owner $\mathcal{I}(d)$ should be encrypted, and the aggregated uploads should statistically satisfy CDP.
\end{itemize}
\end{problem}

\section{Privacy-Preserving FPM Design for Web 3.0}\label{sec_method}

In this section, the Web 3.0-oriented FA-based privacy-preserving FPM scheme, FedWeb, is described in detail.

The FedWeb design is iterative, where the data analysts and data owners communicate for multiple rounds until the FPM task is completed. Each round consists of three phases: 1) \textit{candidate pattern distribution}, where the data analyst distributes possible frequent patterns to the data owners; 2) \textit{DDP-based private response}, where the data owners respond on their received candidates based on the established DDP mechanism; 3) \textit{response analysis}, where the data analyst gathers the uploads, analyzes the uploads, filters candidates, and generates new candidates. In addition, two flexible strategies are proposed to reduce the required data owners.

\subsection{Phase 1: Candidate pattern distribution}\label{subsec_sol1}

In FedWeb, the data analyst maintains a pool of candidate patterns. In the first round, the ``most basic'' candidates, \ie those who cannot be generated by other patterns with the Apriori property, are placed into the candidate pool. For example, in frequent itemset mining tasks, all itemsets with only one item are placed into the candidate pool in the first round. The candidate pool is updated in each response aggregation phase. When the candidate pool becomes empty, the FPM task is completed and the algorithm terminates.

In the candidate distribution phase, the data analyst activates data owners to become participants in this round, and each data owner is assigned a set of candidates to respond to. Two user-defined parameters affect the procedure: the maximal number of candidates each data owner can respond to (response budget) $K$, and the number of responses each candidate receives in one round $P$. Let $\mathcal{C}_t$ denotes the candidate pool in round $t$, and $|\mathcal{C}_t|$ denotes its size. The data analyst should activate sufficient data owners so that each candidate in $\mathcal{C}_t$ has been assigned $P$ data owners, while each data owner can respond to at most $K$ candidates but cannot respond to one candidate more than once. We use $N_t$ to denote the set of participating data owners in round $t$, with size $|N_t|$. 

Each candidate in the candidate pool is given a unique index from 1 to $|\mathcal{C}_t|$. The index is valid throughout this round, and the participating data owners are informed of the indexes of their received candidates. We use $\mathcal{C}_t (i)$ to denote the $i$-th candidate.

\subsection{Phase 2: DDP-based private response}

After receiving at most $K$ candidates from the data analyst, each data owner first checks whether the received candidates are within its local data. Then, it constructs a response vector encoding their response to all the candidates. Next, distributed noise is added to the upload to satisfy DP. Finally, the response vector is uploaded to the data analyst with secure aggregation.

Consider a data owner with index $j$, its local data is denoted $d_j$. The data owners check whether the candidates are within their local data, following the specification of FPM subproblems. If a candidate $\mathcal{C}_t (i)$ is within its local data, it is written as $\mathcal{C}_t (i) \in d_j$. Otherwise, it is $\mathcal{C}_t (i) \notin d_j$.

After the checking, each data owner constructs a response vector with length $|\mathcal{C}_t|$, where each entry represents the data owner's response to a candidate (with the corresponding index). Particularly, the $i$-th entry of the response vector will be 1 when the data owner received candidate $i$ from the data analyst and candidate $i$ is within its local data, or 0 otherwise. 

After the construction of the raw response vector, distributed noise is added to all the received candidates, no matter whether it is within its local data. 
Finally, denoting the final response from data owner $j$ as $R_j$, its $i$-th entry is calculated as follows.
\begin{equation}\label{eq_res}
    R_j[i] =\left\{
    \begin{alignedat}{2}
    &0, &&{\rm~~~if~}j{\rm~does~not~receive~}\mathcal{C}_t (i), \\
    &\mathcal{X} - \mathcal{Y}, &&{\rm~~~if~}j{\rm~receives~}\mathcal{C}_t (i)\land  \mathcal{C}_t (i) \notin d_j,\\
    &1+\mathcal{X} - \mathcal{Y}, &&{\rm~~~if~}j{\rm~receives~}\mathcal{C}_t (i)\land\mathcal{C}_t (i) \in d_,
    \end{alignedat}
    \right.
\end{equation}
where $\mathcal{X}$ and $\mathcal{Y}$ are i.i.d. $P\acute{o}lya(1/P,e^{-\epsilon/K})$ variables. 

\begin{algorithm}[tbp]
 \caption{FedWeb: Data owner procedure}
 \label{algo_owner}
 \begin{algorithmic}[1]
 
 \makeatletter
\def\ALG@special@indent{%
    \ifdim\ALG@thistlm=0pt\relax
        \hskip-\leftmargin
    \else
        \hskip\ALG@thistlm
    \fi
}
\newcommand{\Emptyline}{\item[]\noindent\ALG@special@indent }
\makeatother

\renewcommand{\algorithmicrequire}{\textbf{Input:}}
\renewcommand{\algorithmicensure}{\textbf{Output:}}
\Require Maximal candidates of each data owner: $K$; DDP parameter: $\epsilon$; List of candidates (and their indexes) it should respond $\mathcal{C}$, number of global candidates: $|\mathcal{C}_t|$.
\Ensure Response: $R$
 \Function{Respond}{$K,\epsilon,\mathcal{C}$}
 \State $R \leftarrow$ $|\mathcal{C}_t|$-length all-zero vector
 \For{$c\in\mathcal{C}$ ($i_c$ is the global index of $c$)}
 \If{$c$ is within local data}
 \State $R[i_c] \leftarrow$ 1
 \EndIf
 \State $R[i_c]\leftarrow R[i_c] + (\mathcal{X}-\mathcal{Y})$ \Comment{Eq. \eqref{eq_res} defines $\mathcal{X},\mathcal{Y}$}
 \EndFor
 \State $\mathcal{N} \leftarrow$ randomly picked encryption neighbors
 \For{$x\in\mathcal{N}$}
 \State$\mathcal{M}_x \leftarrow$ pairwise $|\mathcal{C}_t|$-length mask 
 \Emptyline\Comment{Data owner $x$ gets an inversed mask}
 \State $R\leftarrow R + \mathcal{M}_x$
 \EndFor
 \State \textbf{return} $R$
 \EndFunction
\end{algorithmic} 
\end{algorithm}

After the generation of the response vector, the data owners upload the vectors to the data analyst with secure aggregation. In this paper, we utilize the secure aggregation scheme in \cite{bell2020secure}. Each data owner communicates with other $O(\log N_t)$ data owners, and shares a pairwise mask and a secret share of the self mask with each other data owners. After that, each data owner adds the pairwise masks and the self mask to their uploads. The data analyst can learn the sum of the uploads by summing their encrypted uploads to eliminate the pairwise masks, requesting the secret shares from other data owners to eliminate the self mask, and requesting the pairwise masks from other data owners when a data owner drops out. The solution defends collusion among one-third of participating data owners and data analyst, which meets our threat model. Many other options to realize DDP are surveyed in \cite{goryczka2015comprehensive}.

As for now, the workloads of data owners in this round are all completed, and the procedure of data owners in each round (insight derivation function $\mathcal{I}$) is demonstrated in Algo. \ref{algo_owner}, with part of the secure aggregation scheme omitted for the sake of clarity. The details of the omitted part to construct the response vector can be checked in \cite{bell2020secure}. The private response scheme provides a formal privacy guarantee, as shown below.
\begin{theorem}
The private response scheme demonstrated in Algo. \ref{algo_owner} satisfies $\epsilon$-DDP for every data owner.
\end{theorem}
\begin{proof}
The conclusion of Theorem \ref{theorem_ddp} and Definition \ref{def_geo} shows that, the response on each candidate satisfies $(\epsilon/K)$-CDP after it aggregates. Since each data owner responds to at most $K$ candidates, a data owner then satisfies $\epsilon$-DDP according to the composition theorem of DP, which concludes the proof.
\end{proof}

\subsection{Phase 3: Response analysis}

For a candidate $c$ in the candidate pool, the data analyst maintains a profile consisting of three values: the sum of historical response values $r_c$ (represented by the corresponding entry of the response vector), the number of data owners that have responded to the candidate $n_c$, and the number of rounds the candidate stayed in the candidate pool $m_c$ (there is a constant relationship that $n_c = P m_c$). After receiving the aggregated response vector, the data analyst updates the profile of every candidate in the pool. Then, the data analyst performs a two-stage analysis of the candidates: filtering existing candidates and generating new candidates.

For filtering existing candidates, the data analyst evaluates each candidate: whether the candidate can be accepted or rejected as a frequent pattern, and removed from the candidate pool. Since the scheme involves many randomized procedures, we can not guarantee the correctness of every decision. Notice that $r_c/n_c$ is exactly an unbiased estimate of candidate frequency, and the estimation is getting more accurate when the candidate receives more responses from data owners. The decision procedure is driven by the confidence bounds: we accept/reject a candidate as a frequent pattern when there is sufficient probabilistic confidence that its true frequency is higher/lower than the FPM threshold. Otherwise, the candidate remains in the pool to receive more responses in later rounds. 

We derive the final bounds as compounds concluded from the Chebyshev's inequality and the Hoeffding's inequality.
\begin{theorem}[Confidence bound of frequent patterns]\label{theorem_ucb}
A candidate $c$ with profile $(r_c,n_c,m_c)$ is a frequent pattern (target frequency $f$) with confidence $(1-\eta_g)(1-\eta_s)$ when
\begin{equation}
    \frac{r_c}{n_c} - \sqrt{\frac{2e^{-\epsilon/K}}{2(1-e^{-\epsilon/K})^2 P^2 m_c \eta_g}} - \sqrt{\frac{\ln \eta_s}{-2n_c}} \geq f.
\end{equation}
\end{theorem}
\begin{proof}
See Appendix \ref{appendix_ucb}.
\end{proof}

\begin{theorem}[Confidence bound of non-frequent patterns]\label{theorem_lcb}
A candidate $c$ with profile $(r_c,n_c,m_c)$ is not a frequent pattern (target frequency $f$) with confidence $(1-\eta_g)(1-\eta_s)$ when
\begin{equation}
    \frac{r_c}{n_c} + \sqrt{\frac{2e^{-\epsilon/K}}{2(1-e^{-\epsilon/K})^2 P^2 m_c \eta_g}} + \sqrt{\frac{\ln \eta_s}{-2n_c}} \leq f.
\end{equation}
\end{theorem}
\begin{proof}
See Appendix \ref{appendix_lcb}.
\end{proof}

A decision with three options is made for each candidate pattern in the candidate pool: if its profile meets the expression in Theorem \ref{theorem_ucb}, it will be removed from the candidate pool and recorded as a mined frequent pattern; if its profile meets the expression in Theorem \ref{theorem_lcb}, it will be just removed from the candidate pool; otherwise, the candidate will remain in the candidate pool to receive more response from data owners. In particular, in order to save the response resource, a threshold $\tau$ of the maximum response each candidate can receive is set. If a candidate has $n_c \geq \tau$ in any round, it will be forced to be filtered: it will be considered as a frequent pattern if $r_c/n_c \geq f$, or a non-frequent pattern otherwise. Algo. \ref{algo_filter} demonstrates the procedure of filtering one candidate.

After the filtering of existing candidates, the data analyst generates new candidates and places them into the candidate pool to push forward the FPM task. The generation of new candidates leverages the Apriori property of FPM, that the subpatterns of a frequent pattern must be all frequent. To generate new candidates, the data analyst checks the patterns that have been filtered as frequent patterns, and generates a new unexplored pattern as a candidate if all of its subpatterns have been filtered as frequent patterns. For example, in frequent itemset mining tasks, the data analyst will generate a new itemset $\{a,b,c\}$ when the three itemsets $\{a,b\}$, $\{a,c\}$, and $\{b,c\}$ are all filtered as frequent patterns. In frequent sequence mining tasks, a new candidate $a\rightarrow b \rightarrow c$ will be generated when $a\rightarrow b$ and $b\rightarrow c$ are filtered as frequent patterns. For frequent item mining tasks, since all the investigated items are placed in the candidate pool in the first round, no new candidate will be generated in later rounds.

\begin{algorithm}[tbp]
 \caption{Filtering candidate procedure}
 \label{algo_filter}
 \begin{algorithmic}[1]
 
 \makeatletter
\def\ALG@special@indent{%
    \ifdim\ALG@thistlm=0pt\relax
        \hskip-\leftmargin
    \else
        \hskip\ALG@thistlm
    \fi
}
\newcommand{\Emptyline}{\item[]\noindent\ALG@special@indent }
\makeatother

\renewcommand{\algorithmicrequire}{\textbf{Input:}}
\renewcommand{\algorithmicensure}{\textbf{Output:}}
\Require A candidate to be filtered: $c$.
\Ensure Response: Operation on $c$.
 \Function{\textsc{Filter}}{$c$}
 \If{($r_c,n_c,m_c$) satisfies Theorem \ref{theorem_ucb}}
 \State \textbf{return} ``Remove $c$ from $\mathcal{P}$ and add $c$ to $\mathcal{F}$''
 \ElsIf{($r_c,n_c,m_c$) satisfies Theorem \ref{theorem_lcb}}
 \State \textbf{return} ``Remove $c$ from $\mathcal{P}$''
 \ElsIf{$n_c \geq \tau$ and $r_c / n_c \geq f$}
 \State \textbf{return} ``Remove $c$ from $\mathcal{P}$ and add $c$ to $\mathcal{F}$''
 \ElsIf{$n_c \geq \tau$ and $r_c / n_c < f$}
 \State \textbf{return} ``Remove $c$ from $\mathcal{P}$''
 \EndIf
 \State \textbf{return} ``Hold $c$ in $\mathcal{P}$''
 \EndFunction
\end{algorithmic} 
\end{algorithm}

\subsection{Further enhancement on saving response budget}\label{subsec_method_dopolicy}
Although the previously presented mechanism is working, we can observe that some response budgets of data owners are still wasted. Specifically, as each data owner provides $K$ response budgets, when the total number of candidates $|\mathcal{C}_t|$ is less than $K$, the participating data owners cannot use all their budgets within that round. As some budgets are wasted, the system needs more data owners to participate, which harms the effectiveness of FedWeb in Web 3.0 scenarios. Therefore, we enhance our original design with the following two strategies, candidate padding and data owner reusing, to further reduce the data owner usage. 

\textbf{Candidate padding.}
The candidate padding strategy fills the candidate pool with the patterns that are likely to be candidates in the future, until the number of candidates becomes $K$. To realize it, the data analysts sort the existing candidates with their $r_c/n_c$ profiles in decreasing order. Then, the candidates are virtually accepted as frequent patterns one by one. Once a candidate is virtually accepted, some new candidates may be generated based on the Apriori property. These new candidates are put into the candidate pool virtually. When a padding candidate becomes a real candidate in the future, its profile is synchronized.

\textbf{Data owner reusing.} For a data owner that fails to use up its budgets in one round, the data owner reusing strategy asks the data owner to wait for future instructions, instead of completing its participation. Since a data owner must spend its response budgets on different candidates, the awaiting data owners will be reused in later rounds when there is any new candidate that is never responded to by the data owner. The data owner will be reused to respond to more candidates until its budget is used out.

\textbf{Comparison of two strategies.} Data owner reusing performs better in saving response budget, which leads to less data owner usage than candidate padding. But it requires data owners to communicate with the data analyst for many rounds, and keep waiting for the instructions during the procedure. On the other hand, candidate padding only requires one-shot communication between the data analyst and data owners, where the data owner can complete their workload in one round. But it performs less effectively in saving participating data owners.
Data owner reusing is promising in cases like cryptocurrency-related applications where data owners are likely to be always online, but unwilling to share their data freely. Candidate padding is desirable when the data owners frequently drop out, or cannot afford long awaiting, such as the mobile web applications.

The whole procedure of FedWeb has been discussed. The procedure of the data analyst (aggregation function $\mathcal{A}$) is presented in Algo. \ref{algo_analyst}.

\begin{algorithm}[tbp]
 \caption{FedWeb: Data analyst procedure}
 \label{algo_analyst}
 \begin{algorithmic}[1]
 
 \makeatletter
\def\ALG@special@indent{%
    \ifdim\ALG@thistlm=0pt\relax
        \hskip-\leftmargin
    \else
        \hskip\ALG@thistlm
    \fi
}
\newcommand{\Emptyline}{\item[]\noindent\ALG@special@indent }

\makeatother

\renewcommand{\algorithmicrequire}{\textbf{Input:}}
\renewcommand{\algorithmicensure}{\textbf{Output:}}
\Require Maximal candidates of each data owner: $K$; DDP parameter: $\epsilon$; Round-level responder for each candidate: $P$; Response number threshold: $\tau$; Target frequency: $f$.
\Ensure Frequent patterns: $\mathcal{F}$
 \State $\mathcal{F} \leftarrow$ an empty set
 \State $\mathcal{P} \leftarrow$ an empty set \Comment{Set of candidate patterns}
 \State Add basic patterns to $\mathcal{P}$, initialize profile $r_c,n_c,m_c$
 \While{$\mathcal{P}$ is not empty (round $t$)}
 \State $\mathcal{C}_t \leftarrow$ $\mathcal{P}$
 \If{\textit{candidate\_padding\_flag} and $|\mathcal{C}_t| < K$}
 \State Generate new candidates based on Apriori property and place them into $\mathcal{C}_t$, until $|\mathcal{C}_t| = K$
 \EndIf
 \State Generate indexes for candidates in $\mathcal{C}_t$
 \For{$c\in\mathcal{C}_t$}
 \State $D_c \leftarrow$ an empty set \Comment{Data owners responding $c$}
 \EndFor
 \If{\textit{data\_owner\_reusing\_flag}}
 \State Remove saved data owners that used out budgets
 \State Add available saved data owners to $D_c$
 \EndIf
 \State Add new data owners to $D_c$ until $|D_c| = P$ for all $c$
 \State $\mathcal{R}\leftarrow$ $|\mathcal{C}_t|$-length all-zero vector
 \For{each data owner $d$ (new or saved)}
 \State $\mathcal{C}^{<d>}\leftarrow \{c\in\mathcal{C}_t | d\in D_c\}$
 \State $R_d\leftarrow$ \textsc{Respond}($K,\epsilon,\mathcal{C}^{<d>}$)
 \State $\mathcal{R}\leftarrow \mathcal{R} + R_d$
 \If{\textit{data\_owner\_reusing\_flag} and $|\mathcal{C}^{<d>}| < K$}
 \State Save $d$ with its response history
 \EndIf
 \EndFor
 \For{$c\in\mathcal{C}_t$ with index $i_c$}
 \State $r_c \leftarrow r_c + \mathcal{R}[i_c]$; $n_c \leftarrow n_c + P$; $m_c \leftarrow m_c + 1$
 \EndFor
 \For{$c\in\mathcal{P}$}
 \State Operate on $c$ following \textsc{Filter}($c$)
 \EndFor
 \State Generate new candidates form $\mathcal{F}$ with Apriori property
 \State Add new candidates to $\mathcal{P}$, initialize profile $r_c,n_c,m_c$
 \EndWhile
 \State \textbf{return} $\mathcal{F}$
\end{algorithmic} 
\end{algorithm}

\section{Evaluation}\label{sec_evaluation}

We evaluate our design FedWeb with existing local privacy FPM solutions in several web FPM service scenarios, to validate the capacity and effectiveness of FedWeb.

\subsection{Experiment setting}

\textbf{Datasets.} We evaluate FedWeb in three representative web scenarios based on real-world web datasets. 
\begin{itemize}
    \item \textbf{SteemOps (frequent item mining)}: The SteemOps dataset \cite{li2021steemops} is collected from Steemit, a representative decentralized social network application following the Web 3.0 paradigm. In our formulation, each data owner is a Steemit user, and the items in the local data are the other users the data owner has ever reposted to. 
    By completing the frequent item mining task, the data analyst is able to discover these popular users.
    \item \textbf{MSNBC (frequent sequence mining)}: The MSNBC dataset \cite{cadez2000visualization} includes users' browsing trace of news in \textit{msnbc.com}. In our formulation, each data owner holds the browsing history of one MSNBC user, which is essentially a sequence of websites. 
    Mining frequent sequential patterns in the browsing trace, also called user traversal analysis, is a key application for web service analysis.
    \item \textbf{MovieLens (frequent itemset mining)}: The MovieLens dataset \cite{harper2015movielens} includes user's rating in \textit{movielens.org}. In our evaluation, we formulate the local itemset data of users as the themes of movies that a MovieLens user has ever rated full score on them, 
    and completion of the frequent itemset mining helps data analysts to analyze the correlation between different themes of movies. 
\end{itemize}

\textbf{Methods.} In addition to our designed FedWeb, we use three benchmarks for comparison, named FedFPM, SFP, and RAPPOR. We introduce these methods respectively as follows.
\begin{itemize}
    \item \textbf{FedWeb} is the proposed design in this paper. By default, we apply the \textit{data owner reusing strategy} to further reduce the number of participating data owners.
    \item \textbf{FedFPM} \cite{wang2022fedfpm} is a pioneering work in the field of FA-based FPM tasks. It completes multiple FPM tasks under a unified framework. It enforces LDP via a one-bit randomized response, which provides a local privacy guarantee, but leads to a high loss of data utility.
    \item \textbf{SFP} \cite{apple2017learning} is proposed by Apple to discover frequently used words (sequences of letters) from user keyboard inputs. 
    SFP can only handle the sequence data, and is therefore evaluated in the MSNBC dataset only.
    \item \textbf{RAPPOR} \cite{erlingsson2014rappor} is proposed by Google to gather crowd user uploads with an LDP guarantee. 
    It is originally designed for frequent item mining. We modify RAPPOR by encoding all the itemsets as distinct items to solve the frequent itemset mining task. RAPPOR is evaluated in SteemOps and MovieLens datasets.
\end{itemize}

\newcommand{\resfigheight}{0.18\textwidth}
\begin{figure}[!tbp]
    \centering
    \subfigure[F1 score, SteemOps]
    {
        \includegraphics[height=\resfigheight]{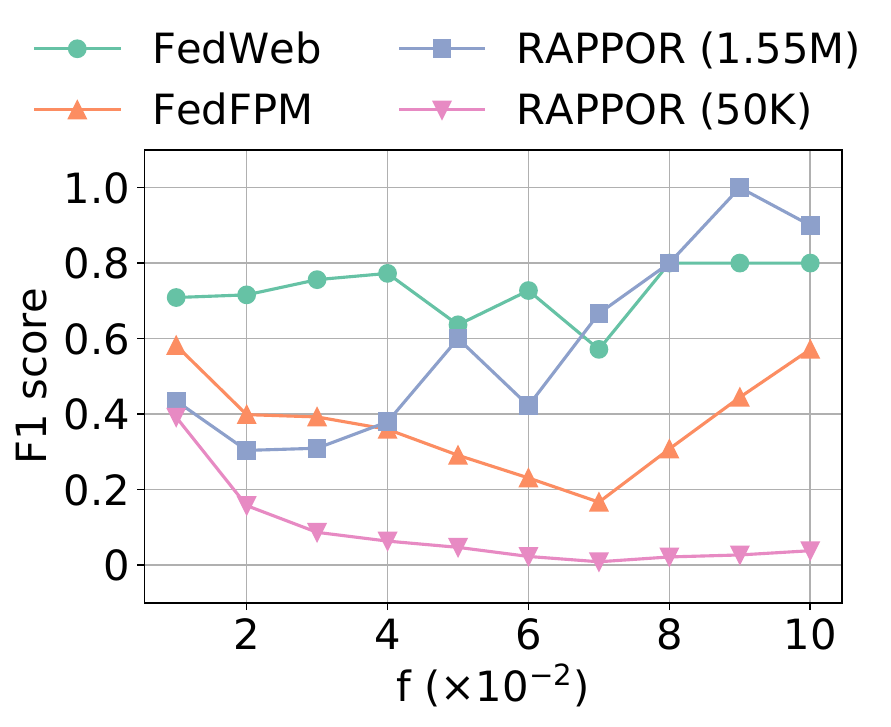}
        \label{subfig_steemops_f1}
    }
    \subfigure[Data owner usage, SteemOps]
    {
        \includegraphics[height=\resfigheight]{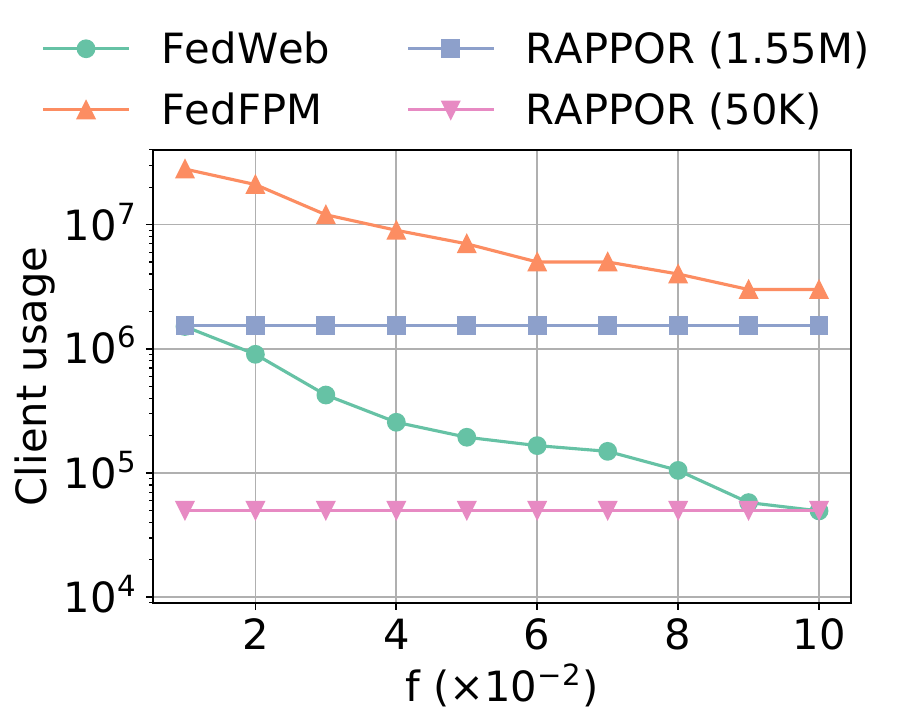}
        \label{subfig_steemops_usage}
    }
    \newline
    \subfigure[F1 score, MSNBC]
    {
        \includegraphics[height=\resfigheight]{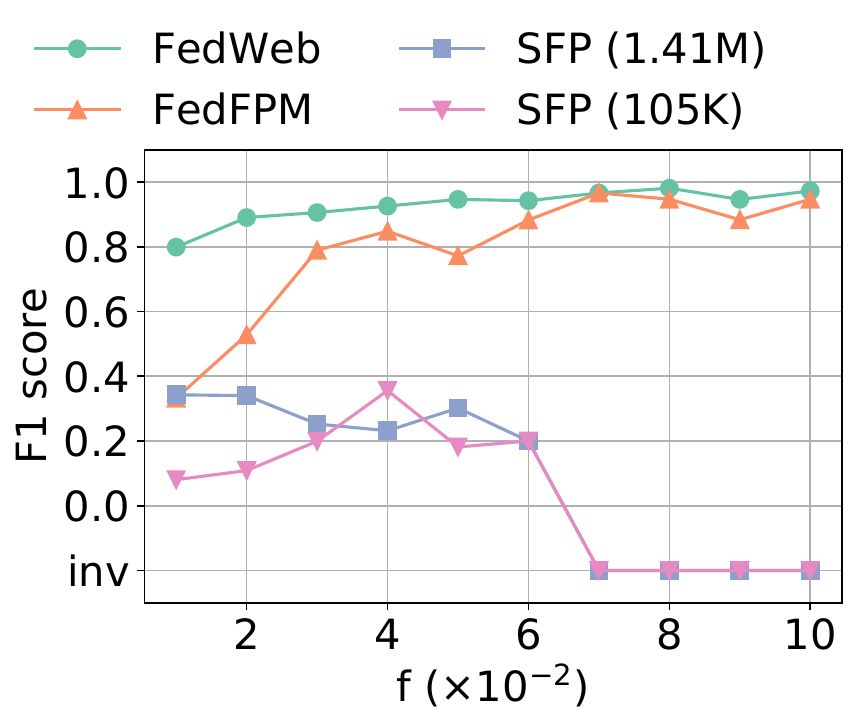}
        \label{subfig_msnbc_f1}
    }
    \subfigure[Data owner usage, MSNBC]
    {
        \includegraphics[height=\resfigheight]{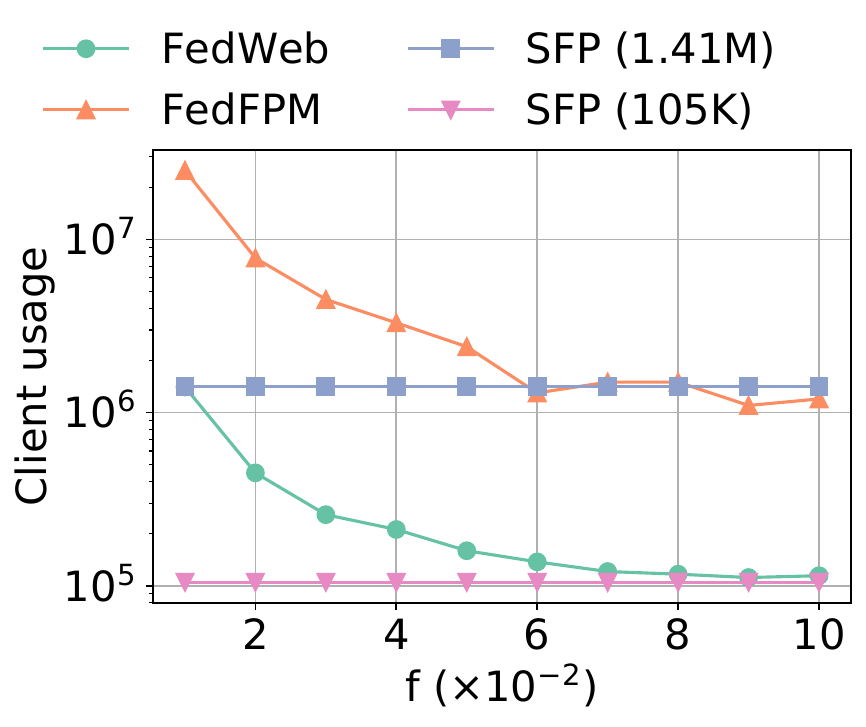}
        \label{subfig_msnbc_usage}
    }
    \newline
    \subfigure[F1 score, MovieLens]
    {
        \includegraphics[height=\resfigheight]{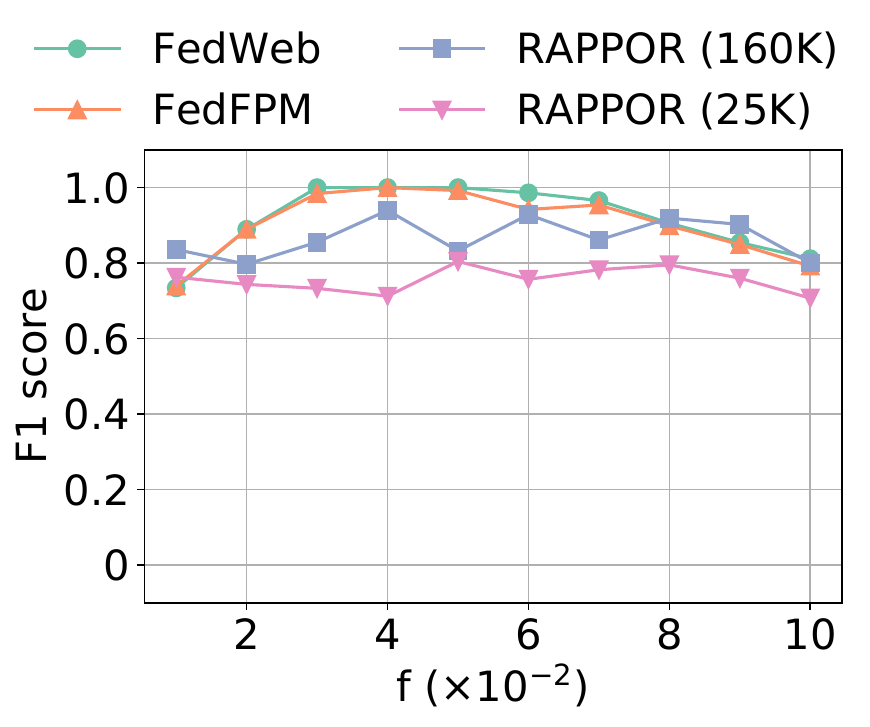}
        \label{subfig_movielens_f1}
    }
    \subfigure[Data owner usage, MovieLens]
    {
        \includegraphics[height=\resfigheight]{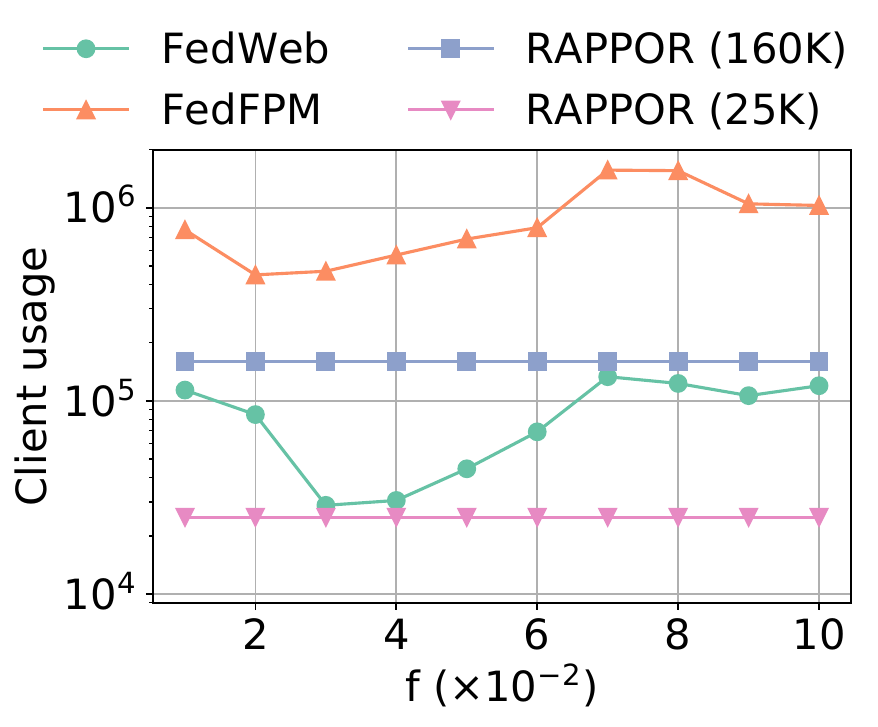}
        \label{subfig_movielens_usage}
    }
    \caption{Performance of FedWeb and the benchmarks in three datasets. The legends of RAPPOR and SFP mark their numbers of participating data owners. We mark invalid (``inv'') if no frequent pattern is output from an algorithm.}
    \label{fig_perf_all}
\end{figure}

\textbf{Metrics.} Under specific DP requirement $\epsilon$, the performance of an FPM algorithm is determined by the data utility (correctness of FPM result, represented by F1 score) and the total number of participating data owners. An effective algorithm should gain high F1 scores using few data owners.

\textbf{Parameters.} We set the target frequencies of FPM to be between 0.01 and 0.1. The number of responses each candidate receives in each round is set to $P=1,000$. The maximal candidate of each data owner is set to $K=50$. The DP parameters of the schemes are all set to $\epsilon=2.0$. The confidence levels used in response analysis are set to $\xi_g=0.01, \xi_s=0.01$. The remaining parameters of FedFPM, SFP, and RAPPOR all follow the default settings in their original papers. As SFP and RAPPOR should predefine the total data owners, we set two values of the data owner number for each RAPPOR/SFP setting: one is slightly higher than the maximal usage of FedWeb, and another is slightly lower than the minimum usage of FedWeb. The first setting represents the case when FedWeb does not gain an unfair advantage by utilizing more data owners, and the latter setting investigates whether the first setting makes the data owner numbers unreasonably high.

\subsection{Experiment results}

\textbf{Performance of FedWeb and benchmarks.} We execute the FPM algorithms in the three scenarios, and record their F1 score of finally mined patterns, and the total number of participating data owners. The results in three datasets are shown in Fig. \ref{fig_perf_all}. FedWeb outperforms RAPPOR, SFP, and FedFPM in F1 scores, even if they utilize more data owners than FedWeb. By averaging the F1 score over the ten target frequencies, FedWeb achieved an improvement in F1 score of at least 25.3\%, 17.3\%, and 1.2\% in SteemOps, MSNBC, and MovieLens datasets, respectively. FedWeb also performs outstandingly in reducing the required data owners. Compared to FedFPM, FedWeb saved 81.1\%$\sim$98.4\% of participating data owners. Lastly, in Figs. \ref{subfig_movielens_f1} and \ref{subfig_steemops_f1}, the decrease of RAPPOR performance with the data owner numbers shows that RAPPOR is still data hungry, and the comparison we make is fair for RAPPOR. On the other hand, the performance of SFP is extremely low, no matter what the data owner number is. It indicates the ineffectiveness of SFP in our tasks.

\begin{figure}[!tbp]
    \centering
    \subfigure[F1 score]
    {
        \includegraphics[height=\resfigheight]{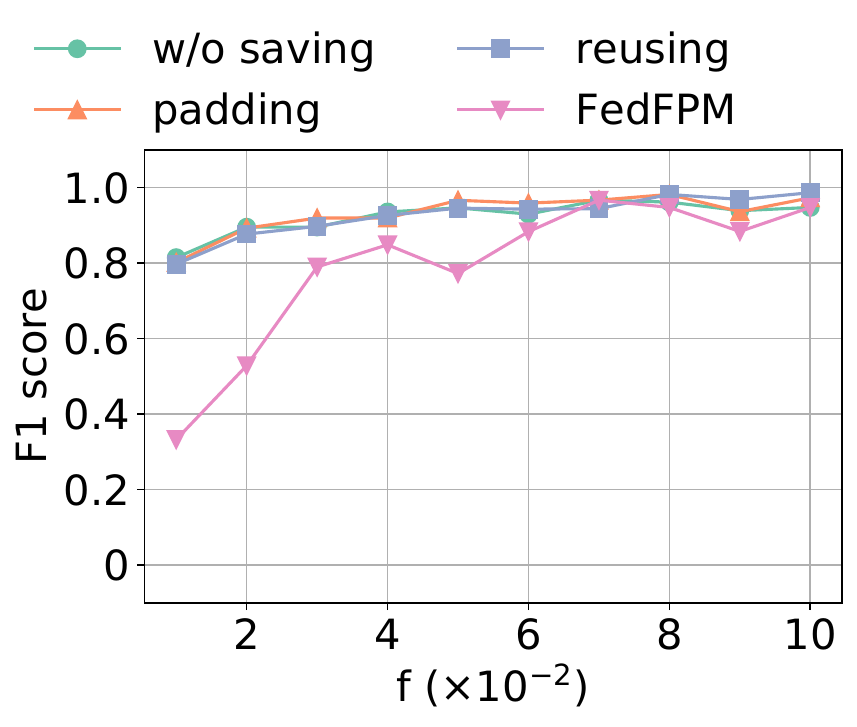}
        \label{subfig_msnbc_dopolicy_f1}
    }
    \subfigure[Data owner usage]
    {
        \includegraphics[height=\resfigheight]{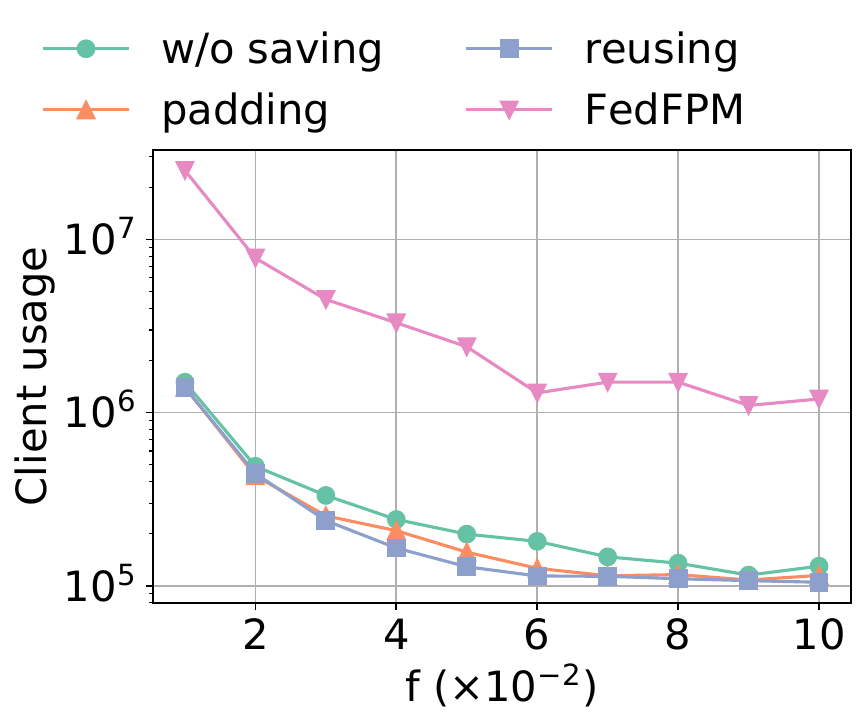}
        \label{subfig_msnbc_dopolicy_usage}
    }
    \newline
    \caption{Performance of FedWeb with different response budget saving strategies in MSNBC dataset. The candidate padding and data owner reusing strategies are abbreviated as ``padding'' and ``reusing'', respectively.}
    \label{fig_perf_dopolicy_msnbc}
\end{figure}

\textbf{Performance of response budget saving strategies.} 
To test the effects of the two response budget saving strategies in Section \ref{subsec_method_dopolicy}, we evaluate the vanilla FedWeb, candidate padding, and data owner reusing in the MSNBC scenario.\footnote{The experiment results similar to Figs. \ref{fig_perf_dopolicy_msnbc} and \ref{fig_perf_heat} conducted in other datasets, which also support our analysis, are omitted due to space limit.} The results are shown in Fig. \ref{fig_perf_dopolicy_msnbc}. Compared to the vanilla FedWeb without further response budget saving, the candidate padding strategy can reduce the total participating data owners by 6.3\%$\sim$30.1\% (16.2\% on average); the data owner reusing strategy can reduce by 6.9\%$\sim$36.9\% (21.6\% on average); and neither of the strategies affects F1 score. Data owner reusing strategy can obtain a better result in saving participating data owners, but it requires data owners to await for multiple rounds. In our experiments, a data owner under the data owner reusing strategy should participate in 3.60 rounds on average. In addition, FedWeb is able to remarkably outperform FedFPM, even without any response budget saving strategy.

\newcommand{\heatfigheight}{0.24\textwidth}
\begin{figure}[!tbp]
    \centering
    \subfigure[F1 score]
    {
        \includegraphics[height=\heatfigheight]{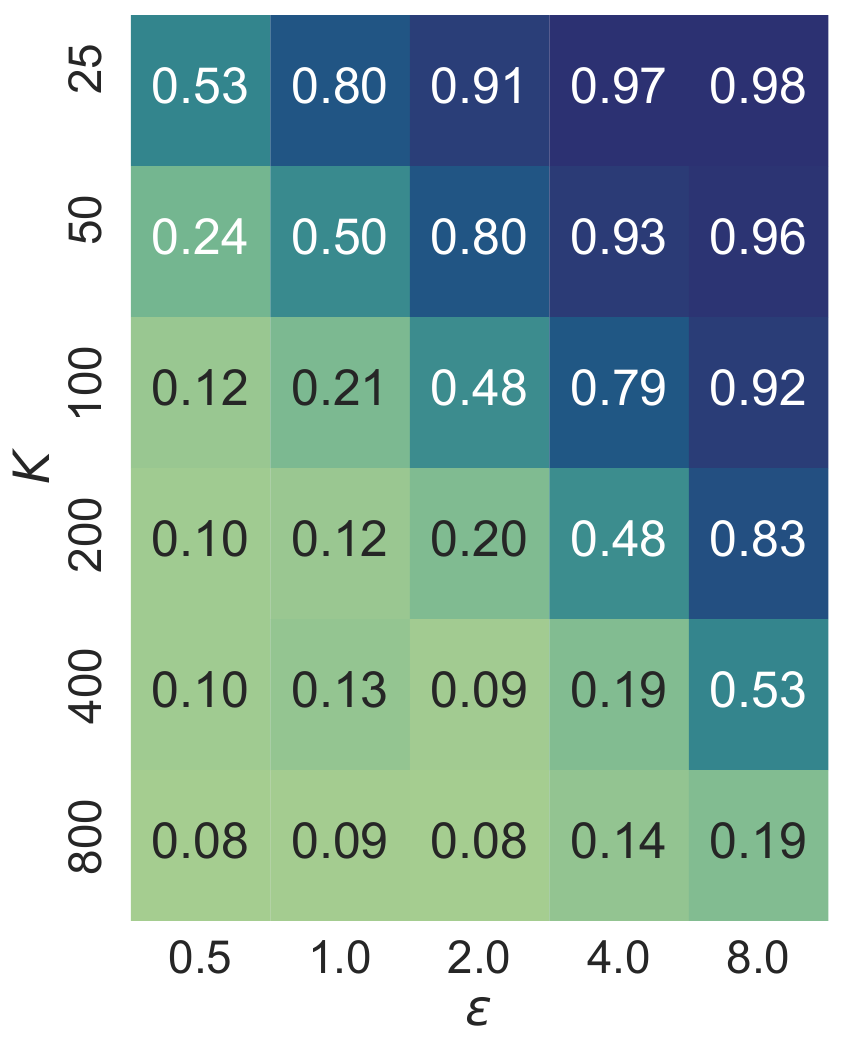}
        \label{subfig_msnbc_heat_f1}
    }
    \subfigure[Data owner usage ($\times 10^5$)]
    {
        \includegraphics[height=\heatfigheight]{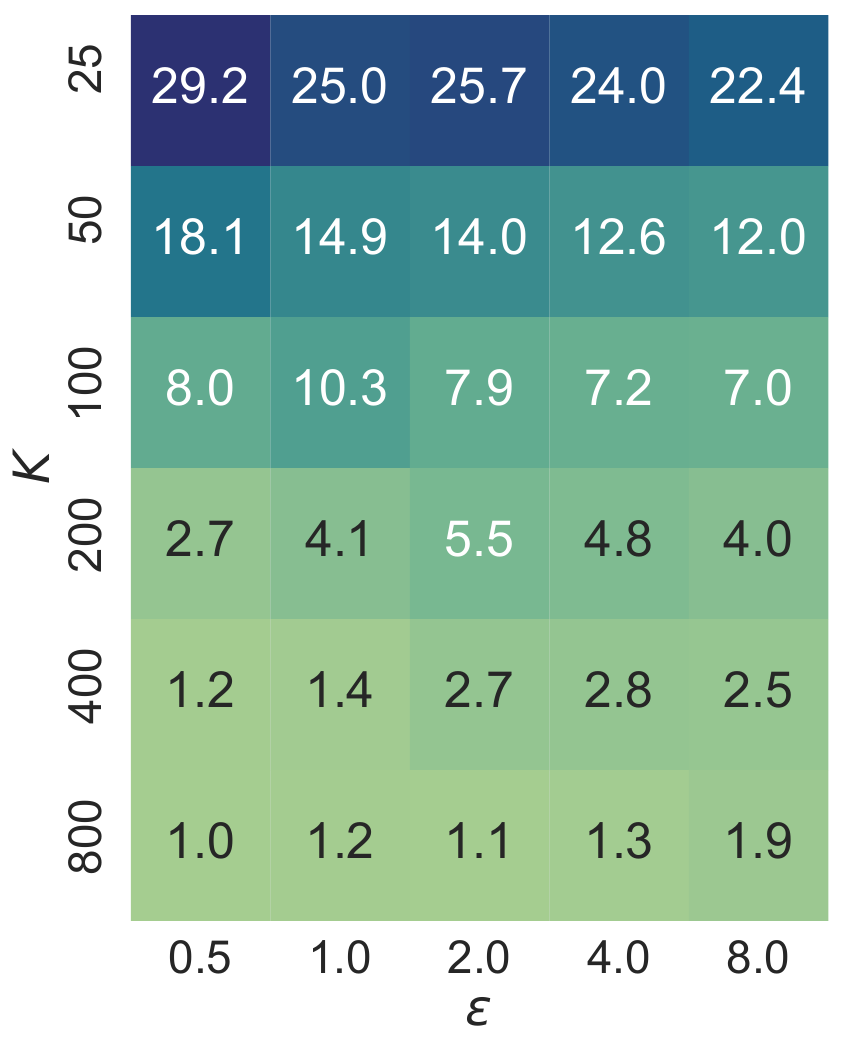}
        \label{subfig_msnbc_heat_usage}
    }
    \caption{Performance of FedWeb under different $\epsilon$ and $K$ in MSNBC dataset.}
    \label{fig_perf_heat}
\end{figure}

\textbf{Parametric study.} 
The DDP parameter $\epsilon$, and the maximum candidate of each data owner $K$ are a pair of important parameters that greatly affect the performance of FedWeb.
Theoretically, a larger $\epsilon/K$ leads to a better data utility, because it determines the noise added to each candidate response, and a larger $K$ leads to fewer data owner usage, because fewer data owners are required when each data owner responds to more candidates. We conduct experiments on the MSNBC datasets with target frequency $f=0.01$ under different $\epsilon$ and $K$ settings, and summarize the resulting F1 scores and data owner usages in Fig. \ref{fig_perf_heat}. By analyzing the F1 score results in Fig. \ref{subfig_msnbc_heat_f1}, we verify our theoretical analysis, that a similar F1 score can be achieved when the value of $\epsilon/K$ is fixed, and the F1 score is higher when $\epsilon/K$ is higher. In Fig. \ref{subfig_msnbc_heat_usage}, we found that when $\epsilon$ is fixed, a higher $K$ can reduce the total required data owners. The parameter $K$ is an important handler for data analysts: when the privacy requirement $\epsilon$ is given, $K$ can be used to tradeoff data utility and data owner usage, where a higher data utility can be achieved by reducing $K$, and a smaller data owner usage can be achieved by increasing $K$.

\section{Related Work}\label{sec_related}

\textbf{Data analytics in Web 3.0.} While the Web 3.0 paradigm emerges, researchers start to place interest in digesting great value from the data generated by Web 3.0. In these works, data from decentralized applications \cite{li2021steemops,DBLP:conf/osdi/LiZZ0XA021} or cryptocurrency blockchains \cite{meiklejohn2013fistful,spagnuolo2014bitiodine,moser2014towards} are gathered for further data analytics or public dataset release. However, they exactly fetch the public data stored in the blockchain, following the least privacy scheme shown in Fig. \ref{subfig_concept_2}. 
Compared to the previous works, this paper presents the first work to introduce FA as a building block in privacy-preserving Web 3.0 data analytics, achieving formal privacy guarantees for web data analytics tasks.

\textbf{Federated analytics.} FA is first introduced by Google to perform data analytics without exposing raw data \cite{fa20}. As for now, FA algorithms are task-specific, where each FA algorithm can be applied to one (or a class of) data analytics problems. The FA model has been applied in various classical and important data analytics tasks, including heavy hitter discovery \cite{zhu2020federated,acharya2019communication}, set computation \cite{pinkas2019spot}, sample mean/median estimation \cite{cormode2021bit,bohler2020secure}, clustering \cite{dennis2021heterogeneity}, and FPM \cite{wang2022fedfpm}. In this paper, we also advance the federated FPM studies by proposing a scalable FA design. 
The idea of utilizing FA for data analytics for Web 3.0 also sheds light on new application scenarios for FA studies. 

\textbf{Privacy-preserving FPM.} There exists a line of works studying FPM under privacy preservation. The early study focuses on resolving one subproblem of FPM, including frequent item mining \cite{erlingsson2014rappor,qin2016heavy}, frequent itemset mining \cite{wang2018locally}, and frequent sequence mining \cite{apple2017learning}. FedFPM \cite{wang2022fedfpm} proposes the first FA framework to unify all FPM subproblems and is the most related work to our problem. However, compared with FedFPM,
FedWeb is Web 3.0-oriented with a special focus on reducing the participating data owners. It utilizes DDP that achieves a formal privacy guarantee without significant data utility loss, derives complex confidence bounds to filter candidates to adapt DDP, and proposes flexible budget saving strategies to further reduce participating data owners.

\section{Conclusion}\label{sec_conclusion}
In this paper, we address the problem of privacy-preserving data analytics on local sensitive data in Web 3.0 systems. 
We propose a reform to the current Web 3.0 data publication scheme, namely, instead of transmitting raw sensitive data, data owners provide FA service to the data analyst. Based on the task-specific FA paradigm, we further design a novel solution for privacy-preserving FPM in Web 3.0 scenarios. The proposed FedWeb mechanism judiciously incorporates DDP into its response scheme with theoretically-sound candidate generation/filtering mechanisms and flexible response budget saving strategies, to achieve reliable privacy preservation, high data utility, and, what is important for Web 3.0, fewer participating data owners. Experiment results show that, our solution can achieve $\sim 25.3\%$ higher F1 score and $\sim 98.4\%$ fewer consumed data owners compared to the benchmarks.

\appendices
\section{Proof of Theorem \ref{theorem_ucb}}\label{appendix_ucb}

Noticing that $\mathcal{X}-\mathcal{Y}$ is with expectation 0, $r_c$ is expected to be increased by 1 if the responding data owner possesses the candidate $c$, and 0 otherwise. Therefore, $r_c / n_c$ is an unbiased estimation of the true frequency of $c$ (denoted $f_c$).

There are exactly two sources of random error on $r_c / n_c$: the \textbf{sampling error} caused by drawing the data owners to respond on a candidate, and the \textbf{geometric error} caused by the geometric noise of the uploads. Denote the two errors as $\mathcal{E}_s$ and $\mathcal{E}_g$, respectively. The relationship between our estimation and the true candidate frequency can be modeled as follows.
\begin{equation}\label{eq_twoerror}
    \frac{r_c}{n_c} = f_c + \mathcal{E}_s + \mathcal{E}_g.
\end{equation}

To start with, we bound the magnitude of $\mathcal{E}_g$. 
$\mathcal{E}_g$ is introduced by the two-sided geometric noise (constructed by distributed P\'olya noises) which is added once in each round. Therefore, $m_c$ noises are added into $r_c$, where each of them follows the distribution in \eqref{eq_geo} with $\alpha=e^{-\epsilon/K}$.

The variance of the added two-sided geometric noise added in each round is calculated.
\begin{equation}
    Var(G(\alpha)) = \frac{2\alpha}{(1-\alpha)^2}.
\end{equation}

By analyzing the essence of $\mathcal{E}_g$, we know that it is 
\begin{equation}
    \mathcal{E}_g = \frac{\sum^{m_c} G(\alpha)}{n_c}=\frac{\sum^{m_c} G(\alpha)}{P m_c}=\frac{\sum^{m_c} \frac{G(\alpha)}{P}}{ m_c},
\end{equation}
exactly the average of $m_c$ random variables following the distribution $\frac{G(\alpha)}{P}$. The variance of $\frac{G(\alpha)}{P}$ is, obviously, $\frac{2\alpha}{P^2 (1-\alpha)^2}$. We then utilize the Chebyshev's inequality \cite{wang2015sublinear} to bound $\mathcal{E}_g$, the average of i.i.d. random variables with known variance.
\begin{theorem}[Chebyshev's inequality for deviation of random variables average] Let $X_1, X_2, ...,X_n$ be i.i.d. random variables with expectation $\mu$ and variance $Var(X)$, we have
\begin{equation}\label{eq_cheby}
    \mathbb{P}\bigg(\frac{\sum_{i=1}^n X_i}{n} - \mu \geq x\bigg) \leq \frac{Var(X)}{2nx^2}
\end{equation}
for any $x>0$.
\end{theorem}
By applying $\mathcal{E}_g$ into \eqref{eq_cheby}, we have
\begin{equation}\label{eq_geoerr1}
    \mathbb{P}(\mathcal{E}_g \geq x) \leq \frac{2\alpha}{2(1-\alpha)^2 P^2 m_c x^2}.
\end{equation}
We set $\eta_g$ as the allowed error rate for geometric error bounding, and recall $\epsilon/K=-\ln \alpha$ to satisfy $\epsilon$-DP, \eqref{eq_geoerr1} can be transformed into
\begin{equation}\label{eq_geoerr3}
    \mathbb{P}\bigg(\mathcal{E}_g \geq \sqrt{\frac{2e^{-\epsilon/K}}{2(1-e^{-\epsilon/K})^2 P^2 m_c \eta_g}}\bigg) \leq \eta_g.
\end{equation}

Then, we try to bound $\mathcal{E}_s$, the error caused by the sampling effect. The bounding is based on an important observation that, $f_c + \mathcal{E}_s$, is equivalent to the frequency of candidate $c$ on a sample of $n_c$ data owners. As a result, $f_c + \mathcal{E}_s$ is an average of $n_c$ random variables, where each variable takes two values 0 or 1 (Bernoulli).

Based on the observation, we use the Hoeffding's inequality \cite{wang2015sublinear} to bound the divergence between $f_c + \mathcal{E}_s$ and $f_c$.
\begin{theorem}[Hoeffding's inequality]
Let $X_1, X_2, ...,X_n\in[0,1]$ be i.i.d. random variables with expectation $\mu$, we have
\begin{equation}
    \mathbb{P}\bigg(\frac{\sum_{i=1}^n X_i}{n} - \mu \geq x\bigg) \leq \exp(-2x^2 n)
\end{equation}
for any $x$.
\end{theorem}
By considering $f_c + \mathcal{E}_s$ as the average of $n_c$ 0-1 bounded variables, we have
\begin{equation}\label{eq_samerr2}
    \mathbb{P}(f_c + \mathcal{E}_s - f_c \geq x) = \mathbb{P}(\mathcal{E}_s \geq x) \leq \exp(-2x^2 n_c).
\end{equation}
We set $\eta_s$ as the allowed error rate for sampling error bounding. Eq. \eqref{eq_samerr2} can be transformed into
\begin{equation}\label{eq_samerr3}
    \mathbb{P}\bigg(\mathcal{E}_s \geq \sqrt{\frac{\ln \eta_s}{-2n_c}}\bigg) \leq \eta_s.
\end{equation}

By summarizing \eqref{eq_geoerr3} and \eqref{eq_samerr3}, we can derive a confidence of $(1-\eta_s)(1-\eta_g)$ that both the expression in \eqref{eq_geoerr3} and \eqref{eq_samerr3} do not hold. After that, we can use \eqref{eq_twoerror} derive a bound of $f_c$ with confidence $(1-\eta_s)(1-\eta_g)$ that
\begin{align}\nonumber
    \mathbb{P}&\bigg(f_c \geq \frac{r_c}{n_c} - \sqrt{\frac{2e^{-\epsilon/K}}{(1-e^{-\epsilon/K})^2 P^2 m_c \eta_g}} - \sqrt{\frac{\ln \eta_s}{-2n_c}} \bigg)\\ 
    &\geq (1-\eta_s)(1-\eta_g).\label{eq_proofbound1}
\end{align}
Eq. \eqref{eq_proofbound1} can be transformed into a confidence bound of $c$ being a frequent pattern when the bound of $f_c$ (RHS of the first line of \eqref{eq_proofbound1}) is larger than $f$, which concludes the proof.

\section{Proof of Theorem \ref{theorem_lcb}}\label{appendix_lcb}

The proof of Theorem \ref{theorem_lcb} follows the same procedure as Appendix \ref{appendix_ucb}. A reversed version of \eqref{eq_geoerr3} can be derived by applying the reversed version of the Chebyshev's inequality.
\begin{equation}\label{eq_geoerr3rev}
    \mathbb{P}\bigg(\mathcal{E}_g \leq -\sqrt{\frac{2e^{-\epsilon/K}}{2(1-e^{-\epsilon/K})^2 P^2 m_c \eta_g}}\bigg) \leq \eta_g.
\end{equation}
Similarly, a reversed version of \eqref{eq_samerr3} can be derived by applying the reversed version of the Hoeffding's inequality.
\begin{equation}\label{eq_samerr3rev}
    \mathbb{P}\bigg(\mathcal{E}_s \leq -\sqrt{\frac{\ln \eta_s}{-2n_c}}\bigg) \leq \eta_s.
\end{equation}
By summarizing \eqref{eq_geoerr3rev} and \eqref{eq_samerr3rev}, considering the joint case that both expressions do not hold, and substituting into \eqref{eq_twoerror}, we can derive a bound similar to \eqref{eq_proofbound1}:
\begin{align}\nonumber
    \mathbb{P}&\bigg(f_c \leq \frac{r_c}{n_c} + \sqrt{\frac{2e^{-\epsilon/K}}{(1-e^{-\epsilon/K})^2 P^2 m_c \eta_g}} + \sqrt{\frac{\ln \eta_s}{-2n_c}} \bigg)\\ 
    &\geq (1-\eta_s)(1-\eta_g).\label{eq_proofbound2}
\end{align}

Eq. \eqref{eq_proofbound2} can be transformed into a confidence bound of $c$ not being a frequent pattern when the bound of $f_c$ is smaller than $f$, which concludes the proof.


\end{document}